%% file: ijcai20main.tex
\documentclass{article}
\usepackage{ijcai20}

\usepackage{times}

\usepackage{soul}
\usepackage{url}
\usepackage[hidelinks]{hyperref}
\usepackage[utf8]{inputenc}
\usepackage[small]{caption}
\usepackage{graphicx}
\usepackage{amsmath}
\usepackage{booktabs}
\usepackage{latexsym}
\usepackage{amssymb,amsfonts}

\usepackage{enumerate}
\usepackage{natbib}

\usepackage{tikz}  
\usepackage{amsthm}
\usepackage{algorithm}
\usepackage{algorithmic}

\renewcommand{\paragraph}[1]{\medskip \noindent {\bf #1.}}

\newcommand*\circled[1]{\tikz[baseline=(char.base)]{
            \node[shape=circle,draw,inner sep=2pt] (char) {#1};}}
\newcommand\numberthis{\addtocounter{equation}{1}\tag{\theequation}}

\newtheorem{definition}{Definition}
\newtheorem{theorem}{Theorem}

\title{The Fairness of Leximin in Allocation of Indivisible Chores}

\author{
Xingyu Chen$^1$\footnote{Now at Facebook}\and
Zijie Liu$^1$\footnote{Contact Author, now at Google}\and
\affiliations
$^1$Duke University\\
\emails
xingyuchen@fb.com, zijieliu@google.com
}

\begin{document}

\maketitle

	\begin{abstract}
		The leximin solution --- which selects an allocation that maximizes the minimum utility, then the second minimum utility, and so forth --- is known to provide EFX (envy-free up to any good) fairness guarantee in some contexts when allocating indivisible goods. However, it remains unknown how fair the leximin solution is when used to allocate indivisible chores. In this paper, we demonstrate that the leximin solution can be modified to also provide compelling fairness guarantees for the allocation of indivisible chores. First, we generalize the definition of the leximin solution. Then, we show that the leximin solution finds a PROP1 (proportional up to one good) and PO (Pareto-optimal) allocation for 3 or 4 agents in the context of chores allocation with additive distinct valuations. Additionally, we prove that the leximin solution is EFX for combinations of goods and chores for agents with general but identical valuations.  		
	\end{abstract}

\input{intro.tex}

\input{preliminary.tex}




\input{prop1.tex}

\input{efx.tex}

\input{conclusion.tex}

\newpage
\bibliographystyle{named}
\bibliography{lit}

\end{document}

%% file: intro.tex
\section{Introduction}
\noindent Fair allocation of indivisible resources plays an important role in the society. In a variety of scenarios, such as task distribution at work, estate partitioning in life, and asset restructuring in business, people seek an allocation that satisfies each party's interest in a fair and efficient manner. To reach this ultimate goal, researchers have come up with various fairness notions \citep{foley1967resource,Lipton2004OnAF,Caragiannis2016TheUF,Conitzer:2017:FPD:3033274.3085125}. However, as \citet{Aziz2016ComputationalSC} noted, most of these fairness notions focus solely on allocation of goods (items with positive utilities), while items that need to be allocated can also be chores (i.e. they have negative utilities). Furthermore, the results in goods allocation may not necessarily carry over to chores. As \citet{aziz2018fair} and \citet{aleks2018envy} demonstrated in their work, several natural extension of the Maximum Nash Welfare solution, which is both fair and efficient when used to allocate indivisible goods \citep{Caragiannis2016TheUF}, fails to provide any reasonable fairness guarantee when chores are involved. The previous observations lead to this key question: does a fair allocation always exist when chores are involved?

In this paper, we partially answer this question by showing that the leximin solution --- which selects an allocation that maximizes the minimum utility, then the second minimum utility, and so forth \citep{sen1976welfare,sen1977social,rawls2009theory} --- provides compelling fairness guarantee in several settings that involve chores.
\subsection{Related work}
\noindent Fair allocation problem has been studied extensively in several different fields and contexts \citep{Brams1996FairD,Aziz2016ComputationalSC}. Among these studies, various different fairness notions have been established  \citep{foley1967resource,Lipton2004OnAF,Budish2010TheCA,Caragiannis2016TheUF,Conitzer:2017:FPD:3033274.3085125}. One popular fairness notion is
\textit{envy-freeness} (EF) \citep{foley1967resource}, which states that no agent should prefer another agent's set of items to her own. Another appealing fairness notion is \textit{proportionality} \citep{steinhaus1948}, which states that each agent should obtain at least $1/n$ of her total value for all the items. However, envy-freeness and proportionality cannot always be guaranteed, especially when we allocate indivisible items. Therefore, a series of relaxations to these notions have been proposed, including \textit{envy-freeness up to one good} (EF1), \textit{envy-freeness up to any good} (EFX) and \textit{proportionality up to one good} (PROP1) \citep{Lipton2004OnAF,Caragiannis2016TheUF,Conitzer:2017:FPD:3033274.3085125}.

Another important concern in allocation problems is efficiency. In this regard, an allocation is said to be \textit{Pareto-optimal} (PO) if there does not exist another allocation that increases the value received by at least one agent without compromising anyone else.

Researchers have also proposed various algorithms which prove the existence of or even efficiently find an allocation that achieves these fairness notions in different contexts.
For fair allocation of indivisible goods, \citet{Lipton2004OnAF} provided an algorithm that finds an EF1 allocation for general valuations but does not guarantee PO. \citet{Caragiannis2016TheUF} showed that the maximum Nash welfare (MNW) solution guarantees EF1 and PO for additive valuations. \citet{Roughgarden2017AlmostEF} showed that the leximin solution guarantees EFX and PO for general but identical valuations with nonzero marginal utility. Meanwhile, the leximin++ solution, a modification of the leximin solution they proposed, guarantees EFX but not PO for all general but identical valuations.

However, fair allocation of indivisible chores only started to receive attention after \citet{Aziz2016ComputationalSC} noted that the results in goods allocation may not necessarily carry over to chores. For combinations of goods and chores, \citet{aziz2018fair} showed that the double round robin algorithm can always generate an EF1 allocation when valuations are additive. For the two agent case, they presented a polynomial-time algorithm that finds an EF1 and PO allocation. \citet{aleks2018envy} also showed that a modified version of the leximin++ solution, which he calls leximax-, finds EFX allocations for indivisible chores with anti-monotone identical valuations. 

For fair allocation of divisible goods and chores, \citet{bogomolnaia2017competitive,bogomolnaia2019dividing} showed that competitive solutions always exist and guarantee EF and PO. When all items are chores, they showed that the efficient allocation with the largest product of disutility is one of possibly many competitive solutions.

\subsection{Our contribution}
\noindent In this paper, we consider the problem of fair and efficient allocation of indivisible chores as well as combinations of indivisible goods and chores for more than two agents. Specifically, we focus on the fairness guarantee provided by the generalization of the leximin solution. 
\begin{itemize}
    \item In Section \ref{sec:modelsandnotions}, we present our model for allocation of indivisible items and extend approximate fairness definitions for goods to include chores. And most importantly, we define our generalization of the leximin solution.
    \item In Section \ref{sec:prop1}, through the leximin solution, we show that a PROP1 and PO allocation of chores always exists for 3 or 4 agents with additive valuations.
    \item In Section \ref{sec:efx}, through the leximin solution, we show that an EFX allocation for combinations of goods and chores always exists for agents with general but identical valuations. With nonzero marginal utility constraint, PO can also be achieved in this context.
\end{itemize}

%% file: preliminary.tex
\section{Model}
\label{sec:modelsandnotions}
Let $\mathcal{N}$ denote the set of agents and $\mathcal{M}$ denote the set of items, where $|\mathcal{N}|=n$ and $|\mathcal{M}|=m$. Throughout this paper, we assume that items are indivisible, so that each item can only be allocated to one agent and in full. Each agent $i$ has a valuation function $v_i: 2^{\mathcal{M}} \rightarrow \mathbb{R}$ that indicates her utility for each subset of $\mathcal{M}$. We slightly abuse notation here and define $v_i(o) = v_i(\{o\})$ for $o\in\mathcal{M}$.

For general valuation function $v$, we assume that $v_i(\emptyset)=0$ and that $v_i$'s are item-wise monotone. If valuation function $v$ is additive, we further assume $v(S) = \sum_{o\in S} v(o)$. By item-wise monotone, we mean that each agent $i$ can partition $\mathcal{M}$ into $(G_i, C_i)$ such that she considers all items in $G_i$ as goods and all items in $C_i$ as chores, i.e. $v_i(S) \leq v_i(S\cup\{g\})$ if $g \in G_i$ and $v_i(S) \geq v_i(S\cup\{c\})$ if $c\in C_i$. Throughout the paper, we slightly abuse the notation of $g$ and $c$. Whenever we use $g$, we implicitly mean that the item is a good; similarly, whenever we use $c$, we implicitly mean that the item is a chore.

Finally, an \textit{allocation} $A$ is a partition of $\mathcal{M}$ into $n$ disjoint subsets, $(A_1, A_2, \ldots, A_n)$, where $A_i$ is the set of items, or \textit{bundle}, given to agent $i$. Throughout the paper, we are only interested in \textit{complete} allocations, where all items are allocated. 

\subsection{Fairness notions}
\label{sec:notion}
In this paper, we make use of two relaxations of envy-freeness and proportionality, namely envy-freeness up to any item (EFX) and proportionality up to one item (PROP1). We present their formal definitions here.
\begin{definition}[EFX]
An allocation $A$ is envy-free up to any item (EFX) if for all $i,j \in \mathcal{N}$ such that $i$ envies $j$, $v_i(A_i) \geq v_i(A_j\backslash\{g\})$ for all $g \in A_j \cap G_i$ and $v_i(A_i) \geq v_i(A_j\cup\{c\})$ for all $c \in A_i \cap C_i$.
\end{definition}
In other words, EFX demands that, if agent $i$ envies agent $j$, removing any item in agent $j$'s bundle that is considered as a good by agent $i$ guarantees that agent $i$ does not envy agent $j$. Similarly, giving agent $j$ a copy of any item in agent $i$'s bundle that agent $i$ considers as a chore guarantees that agent $i$ does not envy agent $j$.

In the same fashion, we give our definition of PROP1 for both goods and chores. 
\begin{definition}[PROP1]
An allocation $A$ is proportional up to one item (PROP1) if for all $i\in\mathcal{N}$,
\begin{itemize}
    \item $v_i(A_i) \geq v_i(\mathcal{M})/n$, or
    \item $\exists g \in \mathcal{M}\backslash A_i$ s.t. $v_i(A_i\cup\{g\}) \geq v_i(\mathcal{M})/n$, or
    \item $\exists c \in A_i$ s.t. $v_i(A_i\backslash\{c\}) \geq v_i(\mathcal{M})/n$.
\end{itemize}
\end{definition}
\hfill \newline
In addition to fairness, people are also concerned with efficiency in allocation problems. A classical efficiency notion is Pareto-optimality.
\begin{definition}
Given an allocation $A$, another allocation $B$ is a Pareto-improvement of $A$ if $v_i(B_i) \geq v_i(A_i)$ for all $i \in \mathcal{N}$ and $v_j(B_j) > v_j(A_j)$ for some $j \in \mathcal{N}$.
\end{definition}
\begin{definition}[PO]
An allocation $A$ is Pareto-optimal (PO) if there is no feasible allocation that is a Pareto-improvement of $A$.
\end{definition}

\subsection{The leximin solution}
\label{sec:aversion}
The traditional leximin solution selects the allocation that maximizes the minimum utility among all agents. If there are multiple allocations that achieve this minimum utility, it chooses among them the one that maximizes the second minimum utility, and so on. To generalize the leximin solution, for each agent $i$, we define an objective function $f_i: 2^{\mathcal{M}} \rightarrow \mathbb{R}^d$ that outputs a $d$-dimensional objective tuple, instead of a scalar utility score. The objective tuples are compared lexicographically. In other words, we start by comparing their first entries. If their first entries are the same, then we compare their second entries, and so on. The generalized leximin solution selects the allocation that maximizes the minimum objective tuple among all agents.

More formally, Algorithm~\ref{alg:generalized_leximin} defines a comparison operator $\prec$ that poses a total ordering over all allocations. Then the leximin solution selects the greatest allocation under this ordering.

\begin{algorithm}[h]
\caption{Leximin Comparison Operator (Returns true if $A \prec B$)}
\label{alg:generalized_leximin}
\begin{algorithmic}
	\STATE $X^A \leftarrow$ order agents by increasing objective tuple $f_i(A_i)$, then by some arbitrary but consistent tiebreak method for agents with the same objective tuple
	\STATE $X^B \leftarrow$ corresponding ordering of agents under B
	\FOR{$l \in \{1, \dots, n\}$}
	\STATE $i \leftarrow X_l^A$
	\STATE $j \leftarrow X_l^B$
	\IF{$f_i(A_i) \neq f_j(B_j)$}
	\RETURN $f_i(A_i) < f_j(B_j)$
	\ENDIF	
	\ENDFOR
	\RETURN \text{false}
\end{algorithmic}
\end{algorithm}
Note that the traditional leximin solution is a special case of this more general leximin solution where $f_i(S) = v_i(S)$. The leximin++ solution defined in \citep{Roughgarden2017AlmostEF}, which maximizes the minimum utility of any agent and then the number of goods owned by the agent with the minimum utility, is another special case of this generalized leximin solution where $f_i(S) = (v_i(S), |v_i(S)|)$.

%% file: prop1.tex
\section{PROP1 and PO allocation of chores for additive valuations}
\label{sec:prop1}
In this section, we show that the leximin solution provides promising fairness and efficiency guarantees in allocation of chores for agents with additive utilities. More formally, suppose that $v_i$'s are additive and re-scale $\pmb{v}$ so that $v_i(\mathcal{M}) = U$ for all $i$ (note that $U$ is negative). Then we have the following result.
\begin{theorem}
    \label{thm:prop1}
	Let $f_i(S) = v_i(S)$. The leximin allocation $A$ with respect to $\pmb{f}$ is PROP1 and PO when there are 3 or 4 agents in the context of chores allocation with additive utilities.
\end{theorem}
\begin{proof}
    PO comes directly from the property of the leximin solution. If we can find an allocation $B$ which is a Pareto-improvement of $A$, then $v_i(B_i) \geq v_i(A_i)$ for all $i\in\mathcal{N}$ and $v_j(B_j) > v_j(A_j)$ for some $j \in \mathcal{N}$ and so $A \prec B$. Thus $A$ cannot be the leximin allocation. 
    
    Suppose \(A\) is not PROP1. According to our definition of PROP1, there must exist an agent $i$ such that 
    \begin{equation}
    	\label{eqn:iprop}
	    v_i(A_i) - v_i(c) < \frac{U}{n}\quad \forall c \in A_i.
    \end{equation}
    Since the leximin allocation $A$ is PO, there must exist an agent \(j\) who does not envy anyone. To see this, we can look at the envy graph generated by allocation $A$. The envy graph contains a directed edge $e = (u, v)$ if agent $u$ envies agent $v$. Suppose for contradiction that every agent \(j\) envies some other agent. Then starting at an arbitrary agent $s$, since every agent envies someone else, we can repeatedly follow an arbitrary outgoing edge from the current agent until we find an envy cycle. For agents in the envy cycle, we can reassign to each agent the bundle she envies. This reassignment will increase the utilities of all agents in the cycle without affecting other agents, which contradicts with PO.
    
    Now consider this agent $j$ who doesn't envy any other agent. Since agent $j$ doesn't envy any other agent, we must have,
    \begin{equation}
    	\label{eqn:ubj}
	    v_j(A_j) \geq \frac{U}{n}.
    \end{equation}
    Let $A_i^-=\{c|c\in A_i, v_i(c) < 0\}$ and $|A_i^-|=a$. Clearly by Equation \ref{eqn:iprop}, $a\geq 2$. Note also $v_i(A_i^-) = v_i(A_i)$.
    
    Consider the chore $c_1$ in $A_i^-$ that agent $j$ dislikes the least, i.e. $c_1={\arg\max}_{c\in A_i^-} v_j(c)$. Suppose for contradiction that \(v_j(A_j) + v_j(c_1) > v_i(A_i)\), we can construct \(A'\) by giving \((A_i\setminus\{c_1\})\) to agent \(i\) and \((A_j\cup\{c_1\})\) to agent \(j\). Since
    $$v_i(A'_i) = v_i(A_i \backslash \{ c_1 \}) = v_i(A_i) - v_i(c_1) > v_i(A_i)$$
    $$v_j(A'_j) = v_j(A_j \cup \{ c_1\}) = v_j(A_j) + v_j(c_1) > v_i(A_i)$$
    $$v_k(A'_k) = v_k(A_k) \quad \forall k \neq \{ i, j \}$$
    we have \(A \prec A'\). This contradicts with the fact that $A$ is the leximin allocation with respect to $\pmb{f}$. So it must be the case that 
    \begin{equation}
    	\label{eqn:ijenvy}
    	v_j(A_j) + v_j(c_1) \leq v_i(A_i).
    \end{equation}
    Consider the chore $c_2$ in agent $i$'s bundle that agent $i$ dislikes the most, i.e. \(c_2=\arg\min_{c\in A_i} v_i(c)\). According to the definition of PROP1, for all $c$ in agent $i$'s bundle, we know \(v_i(A_i) - v_i(c) < \frac{U}{n}\). In particular, we have $v_i(A_i) - v_i(c_2) < \frac{U}{n}$. Suppose \(v_i(A_i)=(1+\epsilon_1)\frac{U}{n}+v_i(c_2)\) where \(\epsilon_1>0\). Since $c_2$ is the most disliked item in agent $i$'s bundle, we have,
    \begin{align*}
        a \cdot v_i(c_2)\leq v_i(A_i)=(1+\epsilon_1)\frac{U}{n}+v_i(c_2).
    \end{align*}
    After some simplification, we have,
    \begin{align*}
        v_i(c_2)\leq\frac{1+\epsilon_1}{a-1}\frac{U}{n}.
    \end{align*}
    Therefore,
    \begin{align*}
        v_i(A_i) &= (1+\epsilon_1)\frac{U}{n}+v_i(c_2)\\
        		 &\leq (1+\epsilon_1)\frac{U}{n} + \frac{1+\epsilon_1}{a-1}\frac{U}{n}\\
        		 &= \frac{a(1+\epsilon_1)}{a-1}\frac{U}{n}\\
        		 &= \left(1+\frac{1}{a-1}+\frac{a\epsilon_1}{a-1}\right)\frac{U}{n}. \numberthis \label{eqn:lbi}
    \end{align*}
    Suppose \(v_i(A_i) = \left(1+\frac{1}{a-1}+\frac{a\epsilon_2}{a-1}\right)\frac{U}{n}\) where \(\epsilon_2 \geq \epsilon_1 > 0\). According to Equation~\ref{eqn:ubj}, we can suppose \(v_j(A_j)=(1-\epsilon_3)\frac{U}{n}\) where \(0 \leq \epsilon_3 \leq 1\). Then from Equation~\ref{eqn:ijenvy} and Equation~\ref{eqn:lbi}, we have,
    \begin{align*}
        v_j(c_1) &\leq v_i(A_i) - v_j(A_j) \\
        		 &= \left(1+\frac{1}{a-1}+\frac{a\epsilon_2}{a-1}\right)\frac{U}{n} - (1-\epsilon_3)\frac{U}{n}\\
        		 &= \left(\frac{1}{a-1}+\frac{a\epsilon_2}{a-1}+\epsilon_3\right)\frac{U}{n}.
    \end{align*}
	Summing agent $j$'s valuation of agent $i$'s bundle and her own bundle, we have,
    \begin{align*}
        v_j(A_i) + v_j(A_j) &\leq av_j(c_1) + v_j(A_j) \\
        &\leq \left[2+\frac{1}{a-1}+\frac{a^2 \epsilon_2}{a-1}+(a-1)\epsilon_3\right]\frac{U}{n}
    \end{align*}
    Therefore, there must exists an agent \( k \neq i, j\) such that,
    \begin{align*}
        v_j(A_k) &\geq \frac{U-v_j(A_i)-v_j(A_j)}{n-2}\\ 
        &\geq \frac{U-\left[2+\frac{1}{a-1}+\frac{a^2 \epsilon_2}{a-1}+(a-1)\epsilon_3\right]\frac{U}{n}}{n-2}\\
        &= \frac{U-2\frac{U}{n}}{n-2} - \left(\frac{\frac{1}{a-1}+\frac{a^2 \epsilon_2}{a-1}+(a-1)\epsilon_3}{n-2}\right)\frac{U}{n}\\
        &= \left(1-\frac{\frac{1}{a-1}+\frac{a^2 \epsilon_2}{a-1}+(a-1)\epsilon_3}{n-2}\right)\frac{U}{n}.
    \end{align*}
    Since agent $j$ doesn't envy any other agent, her valuation of any other agent's bundle cannot be larger than her valuation of her own bundle. However,
    \begin{gather*}
    \frac{1}{a-1}+\frac{a^2 \epsilon_2}{a-1}+(a-1)\epsilon_3 > \frac{1}{a-1}+(a-1)\epsilon_3\\
    \geq \epsilon_3 \cdot \left( (a-1) + \frac{1}{a-1} \right)
    \geq 2\epsilon_3
    \end{gather*}
    Therefore,
    \begin{align*}
    	v_j(A_k) &\geq \left(1-\frac{\frac{1}{a-1}+\frac{a^2 \epsilon_2}{a-1}+(a-1)\epsilon_3}{n-2}\right)\frac{U}{n}\\
    	&> \left(1-\frac{2 \epsilon_3}{n-2}\right)\frac{U}{n}.
    \end{align*}
So \(v_j(A_k) > v_j(A_j)\) when \(n=3\text{ or }4\), which contradicts the assumption that agent \(j\) does not envy anyone. So when \(n=3\text{ or }4\), there can't exist an agent $i$ such that 
\begin{equation*}
v_i(A_i) - v_i(c) < \frac{U}{n}\quad \forall c \in A_i.
\end{equation*}
In conclusion, when there are 3 or 4 agents, the allocation found by the leximin solution with respect to $\pmb{f}$ is PROP1 and PO.
\end{proof}

\subsection{Limitations of the leximin solution}
Since we showed that the leximin solution with $f_i(S) = v_i(S)$ guarantees PROP1 and PO for 3 or 4 agents with additive valuations, a natural question to ask is whether it works for 5 or more agents. Unfortunately, the answer is no. We show this by providing a counterexample with 5 agents, where the leximin allocation is not PROP1.
\begin{theorem}
		Let $f_i(S) = v_i(S)$. The leximin allocation \(A\) with respect to \(\pmb{f}\) cannot guarantee PROP1 when there are 5 or more agents.
\end{theorem}
\begin{proof}
Consider the example shown in Table \ref{tab:counterleximax}. We claim that the leximin allocation with respect to $\pmb{f}$ is marked by the circles.
	\begin{table}
        \addtolength{\tabcolsep}{-5pt}
		\centering
		\begin{tabular}{|c||c|c|c|c|c|c|c|}
			\hline
			& a & b & c & d & e & f & g \\ \hline \hline
			A1 & \circled{-6} & \circled{-6} & \circled{-6} & -9 & -9 & -9 & -10 \\ \hline
			A2 & -18.1 & -18.1 & -18.1 & \circled{-0.1} & -0.2 & -0.2 & -0.2 \\ \hline
			A3 & -18.1 & -18.1 & -18.1 & -0.2 & \circled{-0.1} & -0.2 & -0.2 \\ \hline
			A4 & -18.1 & -18.1 & -18.1 & -0.2 & -0.2 & \circled{-0.1} & -0.2 \\ \hline
			A5 & -18.1 & -18.1 & -18.1 & -0.2 & -0.2 & -0.2 & \circled{-0.1} \\ \hline
		\end{tabular}
		\caption{An example where the leximin solution with respect to $\pmb{f}$ fails to achieve PROP1 with 5 agents}
		\label{tab:counterleximax}
	\end{table}
	\noindent To see this, clearly, we can't allocate any item $a$, $b$ or $c$ to any agent other than Agent 1 since
	\begin{gather*}
	    v_1(\{a, b, c\}) = -18 > v_i(\{o\}) = -18.1\\
	    \forall i\in\{2,3,4,5\},o\in\{a,b,c\}.
	\end{gather*}
	So we must give items $a$, $b$ and $c$ to Agent 1. For the rest of items, it's clear that each agent should get the item they dislike the least. If this is not the case, one agent will have a utility that is no higher than $-0.2$, which is worse than $A$.\\~\\
	However, for Agent 1, 
	\begin{gather*}
	    v_1(A_1 \backslash \{o\}) = -12 < \frac{U}{n} = \frac{-55}{5} = -11 \\
	    \forall o \in A_1 = \{a, b, c\}.
	\end{gather*}
	So the leximin allocation with respect to $\pmb{f}$ is not PROP1 in this example.
\end{proof}

Similar examples can also be constructed for more than 5 agents. This shows that our analysis is in fact tight. The leximin solution works for 3 or 4 agents in this setting but not more than that. \\


%% file: efx.tex
\section{EFX allocation of goods and chores for general but identical valuations}
\label{sec:efx}

\citet{Roughgarden2017AlmostEF} showed that in the context of goods allocation with general but identical valuations, the leximin++ solution guarantees EFX. The leximin++ solution is a special case of the generalized leximin solution where $f(S) = (v(S), |S \cap G|)$ (we omit the $i$ subscripts here since the valuations are identical). 

In this section, we show that by setting $f(S) = (v(S), |S \cap G|, -|S \cap C|)$, the leximin solution with respect to $\pmb{f}$ provides the same fairness guarantee when used to allocate combinations of goods and chores for agents that have generalized but identical valuations. Essentially, the leximin solution with respect to $\pmb{f}$ selects the allocation that maximizes the minimum utility of any agent. After maximizing the minimum utility, it maximizes the number of goods owned by the agent with the minimum utility. Furthermore, it minimizes the number of chores owned by the agent with the minimum utility. If there are multiple allocations that achieve this, this leximin solution repeats this process on the second minimum utility, and so on. Our proof follows the same paradigm as \citep{Roughgarden2017AlmostEF}, but is applied to a more complicated setting.

\begin{theorem}
    \label{thm:efx}
    Let $f(S) = (v(S), |S \cap G|, -|S \cap C|)$. The leximin allocation with respect to $\pmb{f}$ is EFX in the context of allocation of combinations of goods and chores for agents that have general but identical valuations.
\end{theorem}
\begin{proof}
Let $A$ be an allocation that is not EFX. We will show that $A$ cannot be the leximin allocation with respect to $\pmb{f}$.

Since $A$ is not EFX, there exists a pair of agents $i,j$ such that either
$$\exists g \in A_j \cap G \text{ s.t. } v(A_i) < v(A_j \backslash \{ g \})$$
or
$$\exists c \in A_i \cap C \text{ s.t. } v(A_i) < v(A_j \cup \{ c \}).$$
\textbf{Case 1: There exist agents $i,j$ and a good $g \in A_j \cap G$ such that $v(A_i) < v(A_j \backslash \{ g \})$.}

Let $i$ be any agent with utility $\min_k v(A_k)$, the above condition must be true, so we assume without loss of generality that $i = \arg \min_k v(A_k)$. If there are multiple agents with minimum utility in $A$, let $i$ be the one considered last in the ordering $X^A$. Recall that according to Algorithm~\ref{alg:generalized_leximin}, $X^A$ is obtained by ordering agents by increasing objective tuple $f(A_i)$ and then by some arbitrary but consistent tiebreak method.

Define a new allocation $B$ where $B_i = A_i \cup \{g\}$, $B_j = A_j \backslash \{g\}$, and $B_k = A_k$ for all $k \notin \{ i, j\}$. Let $S$ be the set of agents appearing before $i$ in $X^A$. The only bundles that differ between allocations $A$ and $B$ are those of $i$ and $j$, so we have $A_k = B_k$ for all $k \in S$. Thus for all $k \in S$, $v(B_k) = v(A_k) = v(A_i)$. Since $v(B_j) = v(A_j \backslash \{g\}) > v(A_i)$, $j$ must occur after every agent in $S$ in $X^B$.

Since $g \in G$, $v(B_i) = v(A_i \cup \{g\}) \geq v(A_i)$. If $v(B_i) > v(A_i)$, $i$ must occur after every agent in $S$, since $v(B_i) > v(A_i) = v(B_k)$. If $v(B_i) = v(A_i)$, $i$ must still occur after every agent in $S$, since we sort by increasing number of goods if both agents have the same utility.

The utility of every agent besides $i$ and $j$ are the same in $A$ and $B$. We also proved that $i$ and $j$ will still be after $S$ in $X^B$. This means that the first $|S|$ agents in both $X^A$ and $X^B$ are the same, so the leximin comparison won't terminate before position $|S|+1$ in the orderings.

Let $T$ be the set of agents appearing after $i$ in $X^A$. Since $i$ is the last agent with minimum utility, we know
$$v(B_k) = v(A_k) > v(A_i) \quad \forall k \in T \backslash \{ j \}.$$
Now consider the $(|S|+1)$th agent. We know $X^A_{|S|+1} = i$. If $X^B_{|S|+1} = i$, we know $|A_i \cap G| < |B_i \cap G|$, so $A \prec B$. If $X^B_{|S|+1} = k$ for some $k \neq i$, we know $v(A_i) < v(B_k)$, so $A \prec B$ as well. \\~\\
\noindent\textbf{Case 2: There exist agents $i,j$ and a chore $c \in A_i \cap C$ such that $v(A_i) < v(A_j \cup \{ c \})$.}

Define a new allocation $B$ where $B_i = A_i \backslash \{ c \}$, $B_j = A_j \cup \{ c \}$, and $B_k = A_k$ for all $k \notin \{ i, j\}$. Let $S$ be the set of agents appearing before $i$ in $X^A$, $T$ be the set of agents with the same utility, number of goods and number of chores as $i$ in $A$ and $U = S\cup T \setminus \{i\}$. The only bundles that differ between allocations $A$ and $B$ are those of $i$ and $j$, so we have $A_k = B_k$ for all $k \in U$. Thus for all $k \in U$, $v(B_k) = v(A_k) \leq v(A_i)$. Since $v(B_j) = v(A_j \cup \{ c \}) > v(A_i)$, $j$ must occur after every agent in $U$ in $X^B$.

Since $c \in C$, $v(B_i) = v(A_i \backslash \{ c \}) \geq v(A_i)$. If $v(B_i) > v(A_i)$, $i$ must occur after every agent in $U$, since $v(B_i) > v(A_i) \geq v(B_k)$ for all $k \in U$. If $v(B_i) = v(A_i)$, $i$ must still occur after every agent in $U$, since we sort by decreasing number of chores if two bundles have the same utility and the same number of goods.

The utility of every agent besides $i$ and $j$ are the same in $A$ and $B$. We also proved that $i$ and $j$ will both be after $U$ in $X^B$. This means that the first $|U|$ agents in $X^A$ and $X^B$ have the same utility, number of goods and number of chores pairwise. The leximin comparison won't terminate before position $|U|+1$ in the orderings.

Now consider the $(|U|+1)$th agent. Suppose $X^A_{|U|+1} = l$. We know $l$ is the last agent with the same utility, number of goods and number of chores as $i$ in $X^A$. If $X^B_{|U|+1} = i$, we know $v(A_l) \leq v(B_i)$, $|A_l \cap G| = |B_i \cap G|$ and $|A_l \cap C| > |B_i \cap C|$, so $A \prec B$. If $X^B_{|U|+1} = k$ for some $k \neq i$, there are three possibilities since $k$ appears after $l$ in $X^A$, 
\begin{enumerate}
	\item $v(A_l) < v(B_k)$
	\item $v(A_l) = v(B_k)$ and $|A_l \cap G| < |B_k \cap G|$
	\item $v(A_i) = v(B_k)$ and $|A_l \cap G| = |B_k \cap G|$ and $|A_l \cap C| > |B_k \cap C|$.
\end{enumerate}
In all three possibilities, we have $A \prec B$.

Since $A \prec B$ in both cases, $A$ cannot be the leximin allocation with respect to \pmb{f}. Therefore, the leximin allocation with respect to \pmb{f} must be EFX. 
\end{proof}

\subsection{Pareto-optimality}
The leximin solution does not make any guarantee for efficiency in the previous setting. However, if we are willing to make a small assumption about the valuation function, the leximin solution can achieve EFX and PO at the same time. This assumption is on marginal utilities.

Similar to \citep{Roughgarden2017AlmostEF}, we say that a valuation $v$ has nonzero marginal utility if for every set of items $S \subset \mathcal{M}$, $v(S\cup\{g\}) > v(S)$ for all $g \in G \backslash S$ and $v(S\cup\{c\}) < v(S)$ for all $c \in C \backslash S$. If we assume all agents have nonzero marginal utilities, we have the following theorem.
\begin{theorem}
    \label{thm:efx-po}
    Let $f(S) = v(S)$. The leximin allocation with respect to $\pmb{f}$ is EFX and PO in the context of allocation of combinations of goods and chores for agents that have general but identical valuations with nonzero marginal utility.
\end{theorem}
\begin{proof}
Now let $A$ be an allocation that is not EFX. Then there exists a pair of agents $i,j$ such that either
$$\exists g \in A_j \cap G \text{ s.t. } v(A_i) < v(A_j \backslash \{ g \})$$
or
$$\exists c \in A_i \cap C \text{ s.t. } v(A_i) < v(A_j \cup \{ c \}).$$
Assume without loss of generality that $i = \arg \min_k v(A_k)$. If there are multiple agents with minimum utility in $A$, let $i$ be the one considered last in the ordering $X^A$.\\~\\
\textbf{Case 1: There exist agents $i,j$ and a good $g \in A_j \cap G$ such that $v(A_i) < v(A_j \backslash \{ g \})$.} \\~\\
Define a new allocation $B$ where $B_i = A_i \cup \{g\}$, $B_j = A_j \backslash \{g\}$, and $B_k = A_k$ for all $k \notin \{ i, j\}$. Let $S$ be the set of agents appearing before $i$ in $X^A$. Since we assume nonzero marginal utility and $g \in G$, $v(B_i) = v(A_i \cup \{g\}) > v(A_i)$. In addition, we know that $v(B_i) = v(A_j \backslash \{g\}) > v(A_i)$. So $i$ and $j$ must occur after every agent in $S$ in $X^B$. The first $|S|$ agents in $X^A$ and $X^B$ are the same, so the leximin comparison will not terminate before position $|S|+1$ in the orderings.\\~\\
Let $T$ be the set of agents appearing after $i$ in $X^A$. Since $i$ is the last agent with minimum utility, we know
$$v(B_k) = v(A_k) > v(A_i) \quad \forall k \in T \backslash \{ j \}.$$
Now consider the $(|S|+1)$th agent. We know $X^A_{|S|+1} = i$. If $X^B_{|S|+1} = i$, we know $v(A_i) < v(B_i)$, so $A \prec B$. If $X^B_{|S|+1} = k$ for some $k \neq i$, we know $v(A_i) < v(B_k)$, so $A \prec B$ as well. \\~\\
\textbf{Case 2: There exist agents $i,j$ and a chore $c \in A_i \cap C$ such that $v(A_i) < v(A_j \cup \{ c \})$.} \\~\\
Define a new allocation $B$ where $B_i = A_i \backslash \{ c \}$, $B_j = A_j \cup \{ c \}$, and $B_k = A_k$ for all $k \notin \{ i, j\}$. Let $S = \{k|v(A_k) \leq v(A_i), k \neq i\}$. For all $k \in S$, $v(B_k) = v(A_k) \leq v(A_i)$. Since we assume nonzero marginal utility and $c \in C$, $v(B_i) = v(A_i \backslash \{ c \}) > v(A_i)$. In addition, we know that $v(B_j) = v(A_j \cup \{ c \}) > v(A_i)$. So $i$ and $j$ must occur after every agent in $S$ in $X^B$. The first $|S|$ agents in $X^A$ and $X^B$ have the same utility pairwise, so the leximin comparison will not terminate before position $|S|+1$ in the orderings. \\~\\
Now consider the $(|S|+1)$th agent. Suppose $X^A_{|S|+1} = l$. We know $v(A_l) = v(A_i)$. If $X^B_{|S|+1} = i$, we know $v(A_l) < v(B_i)$, so $A \prec B$. If $X^B_{|S|+1} = k$ for some $k \neq i$ and $k \notin S$, we know $v(A_l) = v(A_i) < v(B_k)$, so $A \prec B$ as well. \\~\\
Since $A \prec B$ in both cases, $A$ cannot be the leximin solution. Therefore the leximin solution must be EFX.
\end{proof}

%% file: conclusion.tex
\section{Conclusion and future work}
In this paper we formally studied the fair allocation problem involving both indivisible goods and chores. Our work shows that the leximin solution provides compelling fairness guarantee even when indivisible chores are involved. Some interesting questions for future research are whether EFX/EF1 and PO allocation always exists for 3 or more agents and whether PROP1 and PO allocation always exists for 5 or more agents when the agents have distinct valuation functions.
